\def\draft{0} 

\documentclass[11pt,letterpaper]{article}
\usepackage[margin=1in]{geometry}
\usepackage[utf8]{inputenc}
\usepackage{amsthm}
\usepackage{amsthm, amsmath, amssymb, mathtools}
\usepackage{bbm,bm}
\usepackage{graphicx}
\usepackage{color,xcolor}
\usepackage{paralist,enumitem}
\usepackage{mathrsfs}
\usepackage[normalem]{ulem}

\usepackage[colorlinks=true, allcolors=blue]{hyperref}
\usepackage[capitalise,nameinlink]{cleveref}

\crefname{claim}{Claim}{Claims}

\newcommand{\vnote}[1]{\ifnum\draft=1\textcolor{orange}{[\textbf{Santhoshini:} #1]}\fi}
\newcommand{\mnote}[1]{\ifnum\draft=1\textcolor{red}{[\textbf{Madhu:} #1]}\fi}
\newcommand{\snote}[1]{\ifnum\draft=1\textcolor{red}{[\textbf{Noah:} #1]}\fi}
\newcommand{\sanote}[1]{\ifnum\draft=1\textcolor{red}{[\textbf{Sasha:} #1]}\fi}
\newcommand{\rnote}[1]{\ifnum\draft=1\textcolor{red}{[\textbf{Raghuvansh:} #1]}\fi}
\newcommand{\cnote}[1]{\ifnum\draft=1\textcolor{red}{[\textbf{Chi-Ning:} #1]}\fi}
\newcommand{\pnote}[1]{\ifnum\draft=1\textcolor{red}{[\textbf{Prashanth:} #1]}\fi}

\usepackage{algorithmicx}
\usepackage{algorithm}
\usepackage[noend]{algpseudocode}
\newcounter{algsubstate}

\algnewcommand\algorithmicinput{\textbf{Input:}}
\algnewcommand\Input{\item[\algorithmicinput]}

\algnewcommand\algorithmicoutput{\textbf{Output:}}
\algnewcommand\Output{\item[\algorithmicoutput]}

\algnewcommand\algorithmicgoal{\textbf{Goal:}}
\algnewcommand\Goal{\item[\algorithmicgoal]}

\newcommand{\N}{\mathbb{N}}
\newcommand{\F}{\mathbb{F}}

\newcommand{\ch}{\textsc{Ch}}
\newcommand{\Ch}{\textsc{Ch}}

\usepackage{amsthm}
\usepackage{thmtools,thm-restate}

\numberwithin{equation}{section}
\declaretheoremstyle[bodyfont=\it,qed=\qedsymbol]{noproofstyle}

\declaretheorem[name=Observation,numbered=no]{observation*}

\declaretheorem[numberlike=equation]{theorem}

\declaretheorem[name=Theorem,numbered=no]{theorem*}

\declaretheorem[numberlike=equation]{lemma}
\declaretheorem[name=Lemma,numbered=no]{lemma*}

\declaretheorem[name=Corollary,numbered=no]{corollary*}

\declaretheorem[numberlike=equation]{proposition}
\declaretheorem[name=Proposition,numbered=no]{proposition*}

\declaretheorem[name=Claim,numbered=no]{claim*}

\declaretheorem[name=Conjecture,numbered=no]{conjecture*}

\declaretheorem[name=Question,numbered=no]{question*}

\declaretheoremstyle[bodyfont=\it]{defstyle} 

\declaretheorem[numberlike=equation,style=defstyle]{definition}
\declaretheorem[unnumbered,name=Definition,style=defstyle]{definition*}

\declaretheorem[unnumbered,name=Example,style=defstyle]{example*}

\declaretheorem[unnumbered,name=Notation=defstyle]{notation*}

\declaretheorem[unnumbered,name=Construction,style=defstyle]{construction*}

\declaretheoremstyle[]{rmkstyle}

\title{Algebra in Algorithmic Coding Theory\thanks{This paper accompanies a lecture with the same title given by the author at the workshop titled {\em Forward From the Fields Medal}  at the Fields Institute, Toronto, Canada, August 12-17, 2024.}}

\author{Madhu Sudan\thanks{School of Engineering and Applied Sciences, Harvard University, Cambridge, Massachusetts, USA. Supported in part by a Simons Investigator Award, NSF Award CCF 2152413 and AFOSR award FA9550-25-1-0112. Email: \texttt{madhu@cs.harvard.edu}.}}

\date{December 5, 2025}

\begin{document}

\maketitle

\begin{abstract}
We survey the notion and history of error-correcting codes and the algorithms needed to make them effective in information transmission. We then give some basic as well as more modern constructions of, and algorithms for, error-correcting codes that depend on relatively simple elements of applied algebra. While the role of algebra in the constructions of codes has been widely acknowledged in texts and other writings, the role in the design of algorithms is often less widely understood, and this survey hopes to reduce this difference to some extent.
\end{abstract}

\tableofcontents \newpage

The challenge of error-correction emerged in the late 1930s when the possibility of ``digital'' communication and storage of information started to become a realistic possibility. It was formalized in seminal works of Shannon~\cite{Shannon} and Hamming~\cite{Hamming}.  Digital communication refers to communicating with symbols over a finite alphabet, as we do in most human languages, as opposed to communicating with a continuum of symbols, as one may argue that sound and music do. With digital communication ``perfect communication'', where the receiver is able to receive the sender's message perfectly without any errors, seems plausible, and if achieved, this could ensure that information can be preserved forever. However, no channel of communication is perfect --- errors are inevitable in any setting. And errors in digital communication can have disastrous effects. As an example if the sender were intending to send the message ``WE ARE NOT READY'' and the channel flips just one symbol, it could have the catastrophic effect of the receiver receiving the message ``WE ARE NOW READY''. Error-correction captures the body of work that studies how to ensure that errors of this type do not occur (or at least occur with negligible probability depending on the model of communication). 

Specifically, an error-correction scheme would first ``encode'' the message to be transmitted by adding redundancy and then transmit the encoded data. For example one such scheme might just repeat every symbol a few times. So the actual transmission in the example above may be ``WWWEEE AAARRREEE ...'' by repeating every letter three times. The channel of communication may now introduce errors, but one may hope the errors don't affect all repeated symbols in the same way and so perhaps the receiver may receive a sequence such as ``WWWEXE AAARRYEEE ...''. On seeing a received sequence, potentially with some errors like in the example above, the receiver would try to ``decode'' the sequence to try and retrieve the original message. In the above example the receiver may decide (before seeing the received word) to decode the received transmission by taking the majority symbol in each block of three letters. This would have led to correct decoding (of the first five letters of the message) in the specific example above, but how well does this decoder work in general? Are there better encoders and decoders? For example we could have repeated every character ten times instead of just thrice - what do we gain? And what do we lose? Are there options for encoders other  than just repetition? If so do they come with nice decoders such that both encoders and decoders not only correct many errors at a low ``redundancy'' price, but also can be efficiently computed? These are the questions that frame the rest of this survey. We will start by formalizing these questions in the ensuing section before turning to solutions. 

The one-word summary of the rest of the article is YES! Yes, we can find better encoders and decoders and in fact prove they are close to optimal according to the parameters laid out in the next section. These encoders and decoders will have efficient algorithms accompanying them (and in some cases the only proof that the scheme works well comes from the efficient algorithm)! And all these results are enabled by, mostly elementary, algebra over finite fields.

\section{Error-correction: Problem Definition and Key Parameters}

We start by highlighting some vocabulary that will be used in this article. Many of the terms originate from the work of Hamming~\cite{Hamming}. Some notions, such as the encoding and decoding functions were already mentioned in Shannon~\cite{Shannon}. One concept --- list-decoding --- is from the later work of Elias~\cite{Elias}.

We consider one-way transmission of information from a {\em sender} to a {\em receiver} over a {\em noisy channel}. The channel works over an {\em alphabet} $\Sigma$ which is just some known finite set. 
We will consider the most basic of channels, which outputs an element of $\Sigma$ every time it is asked to transmit one element of $\Sigma$.\footnote{In more general, and many practical settings one considers channels with different input and output alphabets --- we will not consider this more general setting in this paper.} If the channel output equals its input, the transmission is deemed error-free, and deemed an {\em error} otherwise. The quality of the channel is measured by the maximum fraction of symbols that have errors in a finite sequence of transmissions, denoted by $p \in [0,1]$.\footnote{We note that this corresponds to an adversarial source of errors as studied in Hamming's work~\cite{Hamming}, in contrast to a probabilistic source of errors as proposed in Shannon's work~\cite{Shannon} and studied more widely.} We often use $n$ to denote the length of the transmission, and also refer to it as the {\em blocklength}.

The architecture of an error-correcting scheme involves choosing a {\em message space}, and an encoding and decoding function. The {\em encoding function} maps a message to a {\em codeword} over $\Sigma$, where a {\em word} is a finite sequence of symbols of $\Sigma$ and codewords are words obtained by encoding some message. This codeword is transmitted over the channel which produces a {\em received word}, which may differ from the codeword in at most $p$-fraction of symbols. The receiver now applies the decoding function which maps the received word to a small {\em list} (i.e., sequence) of codewords. The transmission is considered successful if the output list contains the message. We note that the classical setting of decoding, which we refer to as {\em unique} decoding, required the list size to be $1$, but we will consider the more relaxed version in this article.

To give the above some set-theoretic formulation, we introduce the notion of Hamming distance. For words $x = (x_1,\ldots,x_n) \in \Sigma^n$ and $y = (y_1,\ldots,y_n) \in \Sigma^n$ we define their (normalized) Hamming distance to be the quantity $\delta(x,y) = \frac1n|\{i \in [n] | x_i \ne y_i\}|$. A $p$-bounded channel $\ch = \{\ch_n\}_{n \in \N}$ is a sequence of functions with $\ch_n:\Sigma^n\to\Sigma^n$ satisfying $\delta(x,\ch_n(x))\leq p$ for every $x \in \Sigma^n$. 

For simplicity (and ease of comparisons) we will assume the message is a word of length $k$ (for positive integer $k$) over the alphabet $\Sigma$ and so the message space is $\Sigma^k$. Thus the encoding is a function $E:\Sigma^k \to \Sigma^n$ and a list of $L$ decoder is a function $D:\Sigma^n \to (\Sigma^k)^L$. A coding scheme is given by an encoder/decoder pair $(E,D)$ and is considered a $(p,L)$-list decoder if for every $m\in \Sigma^k$ and every $p$-bounded channel $\ch$ we have $m \in D(\ch(E(m)))$. (We note that lists conflate the notion of sets and sequences - they are simply meant to be sets, but the decoder outputs the elements of the set in some sequence and this sequence is referred to as a list.) When $L = 1$, the $(p,L)$-list decoder is also called a $p$-error-correction scheme. 

Three key parameters of a coding scheme have already been introduced above: the alphabet size $q = |\Sigma|$, its error-correction bound, i.e., the parameter $p$ above, and the list size $L$. The remaining parameter is the (information) {\em rate} of the coding scheme denoted $R$ which is the ratio $k/n$. 
We summarize the discussion above with the following definition.

\begin{definition}[Coding Scheme]
For positive integers $q,n,L$ and real $R,p \in [0,1]$ an {\em $(R,p,q,L)$-coding scheme of blocklength $n$} is given by a pair $(E,D)$ of functions with $E:\Sigma^k \to \Sigma^n$ and $D:\Sigma^n \to (\Sigma^k)^L$, for some positive integer $k \geq Rn$ and set $\Sigma$ with $|\Sigma|=q$ such that for every $p$-bounded channel $\Ch$ and for every $m \in \Sigma^k$ we have $m \in D(\Ch(E(x)))$. 
\end{definition}

The broad goal is to study all $5$ parameters jointly, but this can be too much, and so we try to reduce to just two as follows. First we only consider families of $(R,p,q,L)$ coding schemes that can achieve arbitrarily large block lengths. Next we make similar choices for $q$ and $L$, i.e., we allow them to also tend to infinity; however, we have to be careful. First we stress that these are not the only settings of interest. Indeed the most common parameters to study are $q=2$ and $L=1$, but that is not the range of parameters that highlight the role of algebra (though the results do imply something in those settings as well and we may comment on this aspect as side notes below). Also when $q$, $L$ and $n$ are all tending to infinity, one needs to be careful about how they grow relative to each other. For instance allowing $L=q^k$ would make coding schemes trivial (as in easy to achieve) and pointless (as in the output list need not give any information about the message being sent). A reasonable choice for the growth of $q$ and $L$ is that each grows as some fixed polynomial in $n$. With these choices the question of attention ends up asking:
\begin{quote}
    For what values of $p$ and $L$ are there polynomials $q(\cdot)$ and $L(\cdot)$ such that for infinitely many blocklengths $n$ there are $(R,p,q(n),L(n))$-coding schemes of blocklength $n$? And how does the answer change if we restrict to polynomial time computable encoding and decoding functions $(E,D)$?
\end{quote}
To spoil the reader's fun, we give the extremely simple answer to this question: 
\begin{quote}
    Such schemes exist if and only if  $R \leq 1-p$, and this remains true even with the restriction that the encoding and decoding functions are polynomial time computable.
\end{quote}

We will get to this result later in this article, but before we do so, we mention one more crucial concept in the study of error-correction, namely the error-correcting code. This object is obtained by considering the image of the encoding function, which is the set of all codewords. The parameter we will focus on here is slightly different and we introduce it below.

A code $C$ of blocklength $n$ over alphabet $\Sigma$ is simply a subset $C \subseteq \Sigma^n$. The minimum distance of a code $C$, denoted $\delta(C)$, is the quantity $\min_{x\ne y \in C} \{\delta(x,y)\}$. The rate of a code is the quantity $R(C) = \frac{\log_q |C|}{n}$ where $q = |\Sigma|$. Note that the image of the encoding function has the same rate as the encoding function itself, provided the encoder is injective. 

Codes of minimum distance $\delta$ also have $(\delta/2,1)$-list decoders. Conversely for every code of minimum distance $\delta$, $\delta/2$ is the best possible value for $\rho$ when restricted to list size $L=1$. 
However allowing larger list sizes sometimes allows us to get better results: In particular, as we will see next, every code of rate $R$ and distance $\delta$ satisfies $R \leq 1 - \delta$. Using unique decoding we can thus correct at most $(1-R)/2$ fraction of errors. But by allowing larger lists we can go all the way to $1-R$ fraction of errors, thus doubling the error-correction capacity at the expense of leaving the receiver with the task of determining which element of the list was the intended message. 

\section{Upper Bounds on the Rate} 

We start with a very simple result dating back to the work of Singleton~\cite{Singleton} that establishes limits on the distance of a code, and the fraction of errors it can correct, conditioned on the rate. 

\begin{theorem}[Upper Bound on Distance and Error-Correction Fraction]\label{thm:php}~
\begin{enumerate}
    \item (The Singleton Bound): For every infinite family of codes $C$ of rate at least $R$ and distance at least $\delta$ we have $\delta \leq 1 - R$.
    \item Further if there are polynomially growing functions $q(\cdot),L(\cdot)$ such that there are infinitely many $n$ for which an $(R,p,q(n),L(n))$-coding schemes of blocklength $n$ exists, then $p \leq 1 - R$.   
\end{enumerate}

\end{theorem}

\begin{proof}
    The proofs are similar and simple applications of the ``pigeonhole principle''.

    For the first part, let $C$ be a code of length $n$, rate $R$ over an alphabet $\Sigma$ of size $q$. So we have $k = \log_q |C| \geq Rn$. Let $\ell$ be the largest integer smaller than $k$. Note $\ell \geq k-1$. Now consider the projection function $\pi:\Sigma^n\to\Sigma^\ell$ that maps $(x_1,\ldots,x_n) \mapsto (x_1,\ldots,x_\ell)$. Let $\pi(C) = \{\pi(x) | x \in C\}$. Since $|\pi(C)| \leq q^\ell < q^k = |C|$ we get that there must exist two distinct codewords $x,y \in C$ such that $\pi(x) = \pi(y)$. For this pair $x,y$, we have $\delta(x,y) \leq \frac{n-\ell}{n}$ (since they can disagree on only the coordinates $i \in \{\ell+1,\ldots,n\}$). We conclude $\delta(C) \leq \delta(x,y) \leq 1 - \frac{\ell}n \leq 1 - \frac{\log_q|C|}n + 1/n = 1 - R + 1/n$. Taking limits as $n \to \infty$ we get $\delta \leq 1 - R$.

    The second part is similar except we need to be slightly more careful with the choice of $n$. Assume for contradiction that there exists $R = 1 - p + \epsilon$ for some $\epsilon > 0$, polynomial $L$ (and arbitrarily growing function $q$) such that for infinitely many $n$ there is a $k \geq Rn$ and $\Sigma$ with $|\Sigma| \leq q(n)$ and an encoding/decoding pair $E,D$ with $E:\Sigma^k \to \Sigma^n$ and $D: \Sigma^n \to (\Sigma^k)^L$ that form a $(p,L)$-list decoding pair. Now assume $n$ is large enough so that $q^{\epsilon n-1} > L = L(n)$. Let $\ell$ be the largest integer that is smaller than $k- \epsilon n$ and let $\pi:\Sigma^n\to\Sigma^\ell$ be the projection function to the first $\ell$ coordinates (as above). Now consider the map $\pi \circ E: \Sigma^k \to \Sigma^\ell$ given by $\pi\circ E(m) = \pi(E(m))$. The average number of preimages under $\pi \circ E$ of an element in $z \in \Sigma^\ell$ is at least $q^{k-\ell} \geq q^{\epsilon n-1} > L$, and so in particular there exists a $z \in \Sigma^\ell$ and messages $m_1,\ldots,m_{L+1}$ such that $\pi(E(m_i))=z$ for every $i \in [L+1]$. Now consider list-decoding $r \in \Sigma^n$ where $r = (z_1,\ldots,z_{\ell},0,0,\ldots,0)$ (where we assume for notational simplicity that $0$ is an element of $\Sigma$). Since $|D(r)| \leq L$ there must exists $i \in [L+1]$ such that $m_i \not\in D(r)$. Now consider transmitting $E(m_i)$ with channel $\ch:\Sigma^n\to\Sigma^n$ that maps the last $n-\ell$ coordinates of the transmitted word to $0$. Since this channel makes at most $n-\ell$ errors, it is a $p$-bounded channel (by our choice of $\ell$ above). Under these choices we have $\Ch(E(m_i)) = r$ and $m_i \not\in D(r)$ which violates the definition of $(p,L)$-decoding scheme.   
\end{proof}

Both proofs above seem extremely weak in proving limits on the power of error-correcting coding schemes and codes. Surely there ought to be a better way of choosing the projections (leave alone using more powerful reasoning) that proves better bounds? This is a tempting thought, but shockingly turns out to be incorrect - the bound above is actually tight. In the case of error-correcting codes this tightness is surprisingly elementary fact as we will see next. For the extension to (algorithmic) coding schemes the result is harder and will form the rest of the article. 

\section{The Reed-Solomon Codes}

We start by describing our first error-correcting code. But before doing that we do a convention switch that will be convenient through the rest of this article. Instead of viewing words as sequences in say $\Sigma^n$, we will view them as functions mapping $[n]$ (or more generally some set of size $n$) to $\Sigma$. In this language a code $C$ is a subset of $\{f: S \to \Sigma\}$ for some set $S$ with $|S|=n$. 

The Reed-Solomon codes are codes where the alphabet is $\F_q$ the finite field with $q$ elements (note such a field exists if and only if $q$ is a power of a prime - this restricts our alphabets accordingly). The encoding function of the Reed-Solomon code views elements of the domain as the coefficients of a univariate polynomial of degree at most $k-1$ (over some formal variable $X$). The encoding function is specified further by a set $S \subseteq \F_q$ of size $n$. (So this restricts $n$ to be at most $q$, which fits our general desire of working with codes over polynomial sized alphabets, but does not lead to codes of arbitrarily long block length over any fixed alphabet.) The encoding now evaluates the polynomial over all elements of $S$. This leads to the following definition.

\newcommand{\RS}{\mathrm{RS}}
\begin{definition}[Reed-Solomon Codes]
    Given $k \leq n \leq q$, a finite field $\F_q$ and $S \subseteq \F_q$ with $|S| = n$, the Reed Solomon encoding function $E^\RS = E^\RS_{k,n,\F_q,S}:\F_q^k \to \{f:S \to \F_q\}$ is given by $E(m) = P_m$ where $m = (m_0,\ldots.m_{k-1}) \in \F_q^k$ and $P_m(\alpha) = \sum_{i=0}^{k-1} m_i \alpha^i$ for $\alpha \in S$. The image of the Reed-Solomon encoding function $E^\RS$ is termed the Reed-Solomon code $C^\RS = C^\RS_{k,n,\F_q,S}$
\end{definition}

From the definition it is clear that the code has length $n$, alphabet of size $q$ and has rate $k/n$. The only remaining parameter to be determined is $\delta(C^\RS)$.

\begin{theorem}
    For every $k \leq n \leq q$ such that the finite field $\F_q$ exists and every $S \subseteq \F_q$ with $|S| = n$, the Reed Solomon code $C^\RS_{k,n,\F_q,S}$ has minimum distance $\delta(C^\RS) = 1 - k/n + 1/n$. 
\end{theorem}

\begin{proof}
    The proof follows immediately from the fact that over any field a non-zero degree $d$ polynomial has at most $d$ roots. Consider two distinct messages $m = (m_0,\ldots,m_{k-1})$ and $m' = (m'_0,\ldots,m'_{k-1})$. Now consider the polynomials $P(X) = \sum_{i=0}^{k-1} m_i X^i$, $P'(X) = \sum_{i=0}^{k-1} m'_i X^i$ and $Q(X) = P(X)-P'(X)$. By definition $Q$ is a non-zero polynomial of degree at most $k-1$. Now for every $\alpha$ the encoding of $m$ and $m'$ agree on the $\alpha$th coordinate, i.e., where $P(\alpha) = P'(\alpha)$ we have $Q(\alpha) = 0$. Applying the aforementioned fact about polynomials we conclude that $E^\RS(m)$ and $E^\RS(m')$ agree in at most $k-1$ places and so disagree in at least $n-k+1$ places. We conclude, that for every $m\ne m'$, we have  $\delta(E^\RS(m),E^\RS(m')) \geq 1 - k/n + 1/n$ and so $\delta(C^\RS) = 1 - k/n + 1/n$.
\end{proof}

Thus we find, not only that there is a code that meets the asymptotic limit $\delta \leq 1 - R$, but it meets the bound $\delta \leq 1 - k/n + 1/n$, exactly for every $k \leq n \leq q$ provided $q$ is a prime power. Over polynomially large alphabets this completely settles the question of the best code with respect to the rate-vs-distance tradeoff. But how does it do for error-correction? And what about the algorithmic efficiency? We discuss these questions next.

First, as pointed out in the paper of Hamming~\cite{Hamming} introducing error-correcting codes, every code of distance $\delta$ has a $(\frac\delta2 - \frac1{2n},1)$-decoder (though this decoder is not necessarily efficient). In the context of Reed-Solomon codes this is equivalent to the statement that given any function $r:S \to \F_q$, there is at most one polynomial $p$ of degree less than $k$ such that 
$|\{\alpha \in S | p(\alpha) \ne r(\alpha)\}| \leq \frac{n-k}2$. And this can easily be seen to be a consequence of the statement that a non-zero polynomial has fewer roots than its degree.\footnote{Specifically, if two distinct polynomials disagree with $r$ on at most $\frac{n-k}2$ points each of $S$, then they disagree with each other on at most $n-k$ of the points of $S$ and so they agree with each other at $k$ points. But this implies that their difference is a non-zero polynomial of degree less than $k$ with at least $k$ roots --- a contradiction.}

So at least as far as unique decoding is concerned and algorithmic considerations are set aside, the Reed-Solomon codes also provide the best possible encoding/decoding scheme. But clearly algorithmic efficiency is crucial to any potential application. (And of course if we can correct more errors than $\delta/2$ that would be good too!) Now the encoding algorithm is already quite efficient - the encoding of a message can be trivially computed in $O(kn)$ time by following the definition and more sophisticated algorithms can solve it in time $O(n\log^2 n)$. The real algorithmic challenge is in decoding. As we will see in the next section, this challenge actually leads to an elegant algebraic question and solutions that also employ algebraic insights.

\section{(List-)Decoding the Reed-Solomon Code: A basic algorithm}

Before giving a solution to the list-decoding problem for \textsc{Reed-Solomon} codes, let us first understand the underlying algorithmic task. Recall that the sender transmits the values of a function $p:S \to \F_q$ where this function is obtained by evaluating a polynomial $p \in \F_q[X]$ of degree less than $k$ at all points of $S$. The receiver knows $S$ and receives a function $r:S\to\F_q$ with the guarantee that $\delta(p,r) \leq p$. In other words $p$ and $r$ agree on at least $t := n(1-p)$ points, where $n = |S|$ and the decoder wished to compute a list of at most $L(n)$ polynomials of degree less than $k$ that includes $p$. We thus get the following problem.

\begin{description}
    \item[Reed-Solomon List-Decoding Problem:]~
    \item[{Input:}] $\F_q,k,t$, $S \subseteq \F_q$ with $|S|=n$ and a function $r:S\to \F_q$. 
    \item[{Output:}] A list $\{p_1,p_2,\ldots\}$ of at most $L$ polynomials of degree less than $k$ that includes every polynomial of degree less than $k$ such that $|\{a \in S \mid p(a) = r(a)\}| \geq t$.
\end{description}

(Note this corresponds to a $(p,L)$-decoder for $p = (n-t)/n$.)

While the problem has algebraic elements, the ``error-locations'' are not algebraic and searching for them in a natural way leads to exponential time algorithms (in either the number of errors, or the number of messages). The solution we describe, from \cite{Sudan97}, will not aim to find these locations directly, but to capture the function $r$ algebraically in a non-trivially nice way. This nice algebraic description, it will turn out, will be able to reveal the information that the decoder seeks. 

In our first attempt we will not aim to optimize the parameter $t$ but get something non-trivially interesting. (Once we understand the idea some optimization is easy and we will do that in the next section.) Our first lemma explains one nice algebraic way to capture the function $r$.

\begin{lemma}\label{lem:biv-interp}
    For every function $r:S \to \F_q$ with $|S|=n$ there exists a non-zero bivariate polynomial $Q(X,Y) \in \F_q[X,Y]$ with $\deg_X(Q),\deg_Y(Q) \leq \sqrt{n}$ such that for all $a\in S$, $Q(a,r(a))=0$. 
\end{lemma}
\begin{proof}
    Expressing $Q(X,Y) = \sum_{i=0}^{\sqrt{n}}\sum_{j=0}^{\sqrt{n}}q_{ij} X^iY^j$ we find that we are trying to find an assignment to $(\sqrt{n}+1)^2> n$ variables $\{q_{ij}\}_{i,j}$ satisfying $n$ homogeneous linear constraints (specifically $\sum_{i=0}^{\sqrt{n}}\sum_{j=0}^{\sqrt{n}}q_{ij} a^ir(a)^j = 0$). Since this is a homogeneous linear system with number of constraints being less that the number of variables, a non-zero solution exists.
\end{proof}

On the one hand the lemma above is ``trivial'' --- it made no assumptions on the function $r$ to prove the existence of $Q$. On the other hand it is algebraically quite powerful since the algebraic degree of $Q$ is just twice $\sqrt{n}$. To see why it helps with list decoding, consider a setting where $t > 2k\sqrt{n}$ and the codeword that the sender sends (before errors) is $p(X)$. In this case the polynomials $Q(X,Y)$ and $R(X,Y) = Y-p(X)$ share $t$ common zeroes in the plane $\F_q \times \F_q$. Specifically every point $a\in S$ such that $r(a)=p(a)$ satisfies $Q(a,r(a))=0 =R(a,r(a))$, where the former equality holds for all $a$ and the latter holds since $r(a) = p(a)$. A classical result from the origins of algebraic geometry, namely B\'ezout's theorem in the plane, tells us that polynomials $Q$ and $R$ of degree $d_1$ and $d_2$ respectively can share more that $d_1d_2$ common zeroes only if they share a common factor. In our setting $d_1 = 2\sqrt{n}$ and $d_2 = k$ and so if $t > 2k\sqrt{n}$ then it follows that $Q(X,Y)$ and $Y-p(X)$ share a common factor. But the latter polynomial is irreducible and so it must be an irreducible factor of $Q(X,Y)$. And this is algorithmically effective since a bivariate polynomial can be factored into irreducible factors in polynomial time over finite fields (and many other settings) due to an illustrious line of work including~\cite{Berlekamp:Factoring-Large,LLL,Kaltofen,Lenstra,grig,GathenG}. 

Since the application of B\'ezout's theorem is something we want to optimize, let us give a self-contained proof in our setting (with already slightly improved parameters). 

\begin{lemma}\label{lem:biv-bezout}
Let $Q(X,Y)\in \F_q[X,Y]$ be a polynomial with $\deg_X(Q) \leq D_x$ and $\deg_Y(Q) \leq D_y$. Further let $p(X)\in\F_q[X]$ be a degree $d$ polynomial. Then if the set of common zeroes of $Q(X,Y)$ and $Y-p(X)$ in $\F_q \times\F_q$, i.e., the set $Z:=\{(a,b)\in \F_q^2 \mid Q(a.b) = b-p(a) = 0\}$, has size greater than $D_x + d\cdot D_y$ then $Y-P(X)$ divides $Q(X,Y)$. 
\end{lemma}
\begin{proof}
    Consider the polynomial $g(X):=Q(X,p(X))$. This is a univariate polynomial of degree at most $D_x + d\cdot D_y$. But for every $(a,b)\in Z$, we have $g(a) = Q(a,p(a)) = Q(a,b) = 0$, where the second equality holds since $b-p(a) = 0$ for $(a,b)\in Z$. Note this also implies that if $(a.b)\ne(a',b')\in Z$ then $a \ne a'$ (since if they were equal we would also have $b=p(a) = p(a') = b'$). Thus the condition on the size of $Z$ implies $g$ has more zeroes than its degree and so $g(X)=0$. 

    Now we view $Q(X,Y)$ as a polynomial in $\F_q(X)[Y]$ (a polynomial in $Y$ with coefficients from the field of rational functions in $X$). We have just seen that substituting $Y=p(X)$ gives us a root of this polynomial and so by the ``division algorithm'' for polynomials over fields we get $Y-p(X)$ is a divisor of $Q(X,Y)$. 
\end{proof}

The resulting algorithm and implications are given below. 

\begin{tabbing}
\quad\= {\bf Basic RS List-Decoding Algorithm}\\
\quad Step 1 \quad\= Compute a non-zero polynomial $Q(X,Y)$ with $\deg_X(Q),\deg_Y(Q)\leq \sqrt{n}$ \\
\quad \quad \quad \quad \quad \quad s.t. $Q(a,r(a))=0$ for all $a\in S$\\
\quad Step 2 \quad\= Factor $Q(X,Y) = Q_1(X,Y)\cdots Q_\ell(X,Y)$ into irreducibles.\\
\quad \quad \quad \quad \quad \quad For every $i\in[\ell]$ if $Q_i(X,Y)$ is of the form $ = Y-p_i(X)$ include $p_i(X)$ in output list. 
\end{tabbing}

\begin{theorem}
    {\bf Basic RS List-Decoding Algorithm} can be implemented to run in polynomial time in $q$. It solves the Reed-Solomon list-decoding problem when $t > k\cdot \sqrt{n}$. 
\end{theorem}

\begin{proof}
    The efficient implementation follows from the fact that Step 1 can be solved by solving a linear system, and Step 2 can employ fast bivariate polynomial factorization algorithms. 

    The correctness follows from the fact, by~\Cref{lem:biv-interp} a polynomial $Q$ satisfying the conditions of Step 1 always exists (and so Step 1 will return such a polynomial); and by \Cref{lem:biv-bezout} any polynomial satisfying the conditions of Step 1 will have $y-p(X)$ as an irreducible factor, provided $t > D_x + D_y\cdot d$ where $D_x = D_y = \sqrt{n}$ and $d = k-1$.  
\end{proof}

\section{Optimizing the List-decoding algorithm: Weighted-degrees and Multiplicities}

The algorithm as analyzed above is non-trivially effective in some ranges of $k$, but not the most important ones when considering the general challenge of error-correction. Specifically it is useful only if $k < \sqrt{n}$ (since otherwise $t$ would need to be larger than $n$ which is not possible). Getting a polynomial time algorithm in this range is not trivial since the natural algorithm would take time exponential in $k$. But in potential application settings we like $k$ to be $\Omega(n)$ so that $R = k/n$ can be positive, and the algorithm above does not work in this regime. But it turns out we can now optimize the algorithm and get significantly better quite easily and get to known limits with a bit more algebra.

But first, let's mention what is known about the limits of $t$ for which the problem may be solvable in polynomial time. Note that to output a list of size $L$ the algorithm must run in time at least $L$ and so the question we focus on is: ``For what value of $t$ is $L$ known to be polynomially bounded in $n$?''. When $t \geq \frac{n+k}2$ we are at the unique decoding limit, i.e., here we know $L=1$. When $t > \sqrt{2kn}$ an inclusion-exclusion argument shows that $L \leq 2n/k$ if $t \geq \sqrt{2nk}$\footnote{Suppose there are $L$ polynomials $p_1,\ldots,p_{L}$ that agree with the received word on $t$ points, then we have that there at $L$ sets $S_1,\ldots,S_{L}$ such $S_i \subseteq S$ with $|S_i| \geq t$ and $|S_i\cap S_j | \leq k-1$ for $i\ne j$. Using an ``inclusion-exclusion argument'', i.e., the sequence of inequalities $n\geq |\cup_i S_i| \geq \sum_i |S_i| - \sum_{i < j} |S_i \cap S_j| \geq Lt - \binom{L+1}2 (k-1)$ one can conclude that $L \leq 2n/k$ if $t \geq \sqrt{2nk}$. In particular we get $L$ is polynomially bounded in $n$ (and even a constant if $n = O(k)$.}.
This new bound is better than (smaller than) the previous one when $k$ is relatively small compared to $n$, but for say $k=n/2$ the previous bound was better. The final improvement, known as the Johnson bound, shows that when $t > \sqrt{kn}$ then $L$ grows polynomially in $n$. When $L$ is larger we really don't know the answer: There are results suggesting that $t \leq \sqrt{kn}$ could really lead to superpolynomial list sizes for some sets $S \subseteq \F_q$ over special fields $\F_q$ (notably the work of Ben-Sasson, Kopparty and Radhakrishnan~\cite{BSKR}). Other recent results show that for random sets $S \subseteq \F_q$ of sufficiently small size, one could imagine going down to $t = k+o(n)$ (see for instance, Alrabiah, Guruswami and Li~\cite{AlrabiahGL2024} and references therein). So beyond the Johnson bound, the answer seems open. In the rest of this section we will show how to achieve the Johnson bound efficiently, i.e., with a polynomial time decoder.

Returning to our goal of reducing the number of agreements needed to perform list-decoding, we start with a very simple observation that exploits the imbalance between the role of $D_x$ and $D_y$ in \Cref{lem:biv-bezout}. Clearly, to reduce $t$, it is preferable to reduce $D_y$ even at a cost of somewhat larger $D_x$. All we need is that the number of valid monomials in the support of $Q$ be strictly larger than $n$ to make \Cref{lem:biv-interp} work. So if we let $D_y$ be a parameter and set $D_x = \lceil (n+1)/D_y \rceil$ then the proof of \Cref{lem:biv-interp} can be modified to prove that a non-zero polynomial $Q$ with $\deg_X(Q) \leq D_x$ and $\deg_Y(Q)\leq D_y$ exists satisfying $Q(a,r(a))=0$ for all $a \in S$. And then  \Cref{lem:biv-interp} shows that if $t > D_x + (k-1)D_y$ then the list-decoding problem can be solved by factoring $Q$. Setting $D_y \approx \sqrt{n/k}$ and $D_x \approx \sqrt{nk}$ then yields a solution to the Reed-Solomon List-Decoding Problem for $t > 2\sqrt{kn}$. This is already gets us to ranges where $k = \Omega(n)$ but not up to the Johnson bound, or even the inclusion-exclusion bound.

But we can optimize the Basic Algorithm further in two steps. The first focuses on the monomials in the support of $Q$ more carefully using the notion of weighted degree.

\subsection{Weighted Degrees}

\begin{definition}
    For positive integers $d_x$ and $d_y$, the $(d_x,d_y)$-weighted degree of the monomial $X^i Y^j$ is the quantity $id_x + jd_y$.
    The $(d_x,d_y)$-weighted degree of a polynomial $Q(X,Y) = \sum_{i,j}q_{ij}X^iY^j$ is the maximum over all monomials $X^iY^j$ with $q_{ij}\ne 0$ of the $(d_x,d_y)$-weighted degree of the monomial $X^i Y^j$.
\end{definition}

The following lemmas are easy adaptations of \Cref{lem:biv-interp} and \Cref{lem:biv-bezout} to the setting of weighted degrees.

\begin{lemma}\label{lem:wt-biv-interp}
    For every function $f:S \to \F_q$ with $|S|=n$ there exists a non-zero bivariate polynomial $Q(X,Y) \in \F_q[X,Y]$ with $(1,k-1)$-weighted degree at most $\sqrt{2kn}$ such that for all $a\in S$, $Q(a,r(a))=0$. 
\end{lemma}

\begin{lemma}\label{lem:wt-biv-bezout}
Let $Q(X,Y)\in \F_q[X,Y]$ be a polynomial with $(1,k-1)$-weighted degree at most $D$. Further let $p(X)\in\F_q[X]$ be a degree $k-1$ polynomial. Then if the set of common zeroes of $Q(X,Y)$ and $Y-p(X)$ in $\F_q \times\F_q$, i.e., the set $Z:=\{(a,b)\in \F_q^2 \mid Q(a,b) = b-p(a) = 0\}$, has size greater than $D$ then $Y-P(X)$ divides $Q(X,Y)$. 
\end{lemma}

Now define the Weighted Degree Reed Solomon List-Decoder to be the modification of the Basic Reed Solomon List-Decoder with the only change being that in Step 1 it finds a polynomial $Q$ of $(1,k-1)$-weighted degree at most $\sqrt{2kn}$ that is zero on points $(a,r(a))$. From \Cref{lem:wt-biv-interp} and \Cref{lem:wt-biv-bezout} we immediately get the following theorem.

\begin{theorem}
    {\bf Weighted Degree RS List-Decoding Algorithm} can be implemented to run in polynomial time in $q$. It solves the Reed-Solomon list-decoding problem when $t > \sqrt{2kn}$. 
\end{theorem}

This now gives an efficient list-decoder up to the inclusion-exclusion bound. In fact by carefully playing with the weighted degree (and modifying \Cref{lem:wt-biv-interp} appropriately), we can even get a ``perfect'' unique decoder (one that corrects $(n-k)/2$ errors), by setting the weighted degree bound to $(n+k)/2$, and restricting $Q(X,Y)$ to be of the form $A(X)\cdot Y + B(X)$. Since such a polynomial can only have one factor of the form $Y-p(X)$ the output list has size $1$ as required. This algorithm turns out to be equivalent to an algorithm due to Berlekamp and Welch~\cite{WelchB} from the 1980s (in particular following the exposition in \cite{GemmellS}). 

\subsection{Multiplicities}

The last step towards algorithmic achievement of the Johnson bound, i.e., improving requirement on the agreement parameter from $t > \sqrt{2kn}$ to the weaker $t > \sqrt{kn}$ turns out to involve a significantly more sophisticated algebraic idea, namely multiplicity of zeroes --- a notion with a rich history in algebraic geometry, and its applications in adjacent fields. We aim to present some self contained definitions below, though we will skip many of the proofs.

\begin{definition}[Multiplicity of a root for univariate polynomials]
We say that a polynomial $P(X)\in\F_q[X]$ has a root of multiplicity $m$ at $a\in \F_q$ if $P(X)$ is divisble by $(X-a)^m$. 
\end{definition}

We really need the notion extended to multivariate polynomials but based on our understanding of roots of multiplicity $1$, it should be clear that the condition will not be based on the factorization of the bivariate polynomial. So we first note some equivalent ways of expressing the multiplicity of a root of a univariate polynomial. Note that the condition above is equivalent to saying $X^m$ divides the polynomial $P(X+a)$ - but unfortunately this is still a condition about the divisors of a polynomial. So we reformulate this again to say that $a$ is a root of multiplicity $m$ of $P(X)$ if the polynomial $P(X+a)$ is not supported on monomials of degree less than $m$. This condition now extends naturally to bivariate polynomials as below.

\begin{definition}[Multiplicity of a root for bivariate polynomials]
We say that a polynomial $Q(X,Y)\in\F_q[X,Y]$ has a root of multiplicity $m$ at $(0,0)$ if $Q(X,Y)$ is not supported on any monomial of degree less than $m$. (I.e., if $Q(X,Y) = \sum_{ij} q_{ij}X^iY^j$ and $i+j < m$ then $q_{ij}=0$.) 
We say that $Q(X,Y)$ has a root of multiplicity $m$ at $(a,b)\in\F_q \times \F_q$ if $Q(X+a,Y+b)$ has a root of multiplicity $m$ at $(0,0)$.
\end{definition}

For reasons that are not completely obvious (and require following the numbers carefully) it turns out that rather finding the ``first'' polynomial $Q$ that is zero on the given set of points, namely $\{(a,r(a))\mid a\in S\}$ it is better to find a polynomial that has a zero of high multiplicity at all of the above points, even though this forces the degree of $Q$ to be higher. Such polynomials, it turns out, need fewer intersections with the polynomial $Y-p(X)$ to have the latter as a factor. Again we state lemmas without proofs that quantify these effects.

\begin{lemma}\label{lem:mult-biv-interp}
    For every function $f:S \to \F_q$ with $|S|=n$ and every positive integer $m$,  there exists a non-zero bivariate polynomial $Q(X,Y) \in \F_q[X,Y]$ with $(1,k-1)$-weighted degree at most $m\cdot \sqrt{(1+\frac1m) kn}$ such that $Q(X,Y)$ has a zero of multiplicity $m$ at $(a,r(a))$ for every $a\in S$.
\end{lemma}

(While we skip the proof, a hint is that this follows from the fact that requiring some point $(a,b)$ to be a root of multiplicity $m$ of $Q$ imposes  $\binom{m+1}2$ linear constraints on the coefficients of $Q$.) 

\begin{lemma}\label{lem:mult-biv-bezout}
Let $Q(X,Y)\in \F_q[X,Y]$ be a polynomial with $(1,k-1)$-weighted degree at most $D$. Further let $p(X)\in\F_q[X]$ be a degree $k-1$ polynomial. Then if the set of zeroes of $Y-p(X)$ contains more than $D/m$ points in common with the roots of multiplicity $m$ of $Q$, $Y-p(X)$ divides $Q(X,Y)$. 
\end{lemma}

Combining the two lemmas together and letting $m$ grow large enough gives the following theorem (due to Guruswami and Sudan~\cite{GuruswamiS98}). 

\begin{theorem}\label{thm:rs-gs}
    {\bf Weighted Degree RS List-Decoding Algorithm} can be implemented to run in polynomial time in $q$. It solves the Reed-Solomon list-decoding problem when $t > \sqrt{kn}$. 
\end{theorem}

\section{Achieving ``capacity'': Folded-Reed Solomon Codes}

While the final theorem above (\Cref{thm:rs-gs}) achieves the best known bounds on the list-decodability of Reed-Solomon codes, they are not optimal in terms of their relationship between their rate and list-decoding radius (the fraction of errors that can be list-decoded with polynomial sized lists). The following theorem describes this optimal behavior for codes:

\begin{theorem}\label{thm:comb-list-dec-cap}
    For every $p \in [0,1]$ and $\epsilon > 0$ the following hold:
    \begin{enumerate}
        \item There exists a $q_0 = q_0(\epsilon)$ and a polynomial $L_\epsilon(\cdot)$ such that for every $q \geq q_0$ and every sufficiently large $n$, there is an alphabet $\Sigma$ with $|\Sigma|=q$ and $k \geq (1-p-\epsilon)n$ and an encoding function $E:\Sigma^k \to \Sigma^n$ that is $(p,L_\epsilon(n))$-list-decodable.
        \item For every polynomial $L(n)$ and every sufficiently large $n$ and for every $q$ if an encoding function $E:\Sigma^k\to\Sigma^n$ has a $(p,L(n))$-list-decoder then $k \leq (1-p+\epsilon)n$. 
    \end{enumerate}
\end{theorem}

We omit the proof of the theorem but note the following: The first part can be proved relatively easily by the probabilistic method, specifically, by picking $E$ at random and proving that the code is $(p,O(n))$-list-decodable with high probability (and so in particular, a code like this exists). And the lower bound (Part 2) is immediate from Part (2) of \Cref{thm:php}.

In simpler terms the theorem above says that encoding functions of rate strictly below $1-p$  that can be list-decoded from $p$-fraction of errors with polynomial sized lists exist over sufficiently large (constant-sized) alphabets. And such functions of rate greater than $1-p$ do not exist, over any alphabet! Thus $1-p$ is the limit of the rate for list-decoding from $p$ fraction errors. This limit is called the {\em list-decoding capacity} for $p$-fraction errors, and families of codes approaching this rate in the limit are called capacity-achieving codes.\footnote{We note that the term ``capacity'' has a rich and loaded history in information theory, and is traditionally applied to channels that inject error stochastically. In this section we are applying this term to a channel that injects error adversarially with a relaxed notion of decoding, namely list-decoding. Despite the differences we believe the use of the phrase is appropriate and consistent with the spirit of the term as used in information theory.} 

However Part (1) of the theorem above is non-algorithmic in two aspects: First, the encoders guaranteed to exist are not explicit and in particular not given by a polynomial time computable function. Next, the decoders are also not polynomial time computable. Achieving either of these features remained an open question till 2006, until a breakthrough work of Guruswami and Rudra~\cite{GuRu} in 2006 (which in turn built on a brilliant advance of Parvaresh and Vardy~\cite{ParvareshV}). In subsequent years the analysis of these codes has been simplified further and we follow the analysis of Guruswami and Wang~\cite{GuruswamiW13}.

\paragraph{The ideas:} The key element that allowed the list-decoder for Reed-Solomon codes to work is that by using two variables, one for the coordinates of the code, and one for the values of the codewords at the coordinate, we could get a non-trivially low-degree polynomial $Q(X,Y)$ that explains the graph of the function $(a,f(a))_{a \in S}$. Parvaresh and Vardy, in turn inspired by Coppersmith and Sudan~\cite{CoSu}, suggested adding more information about $f$ to the encoding. For instance, one could send along values of $g(a) = f^2(a)$ in addition to sending $f(a)$ for every $a \in S$. Or perhaps we could use $g(a) = f'(a)$ where $f'(X)$ is the derivative of $f(X)$. Why might this help? We now we get to find a trivariate polynomial $Q(X,Y,Z)$ such that $Q(a,f(a),g(a))=0$ for every $a \in S$ and the addition of the new variable reduces the degree even further (to something growing like $\sqrt[3]{n}$). However one has to be careful in doing so: The polynomial $Q$ is forced to find some simple relationship among the three coordinates of the input triples, but might find a useless one. For example if we used $g(X) = f(X)^2$ and there were no errors in communication, then the polynomial $Q(X,Y,Z)$ might simply be $Z - Y^2$ and this would be zero everywhere and contain no information about $f$! Finding the right function $g$ (and later we'll select a whole sequence of functions $g_1,g_2,\ldots$) turns out to be non-trivial. Eventually the choice we will make is to use $g_1(X) = f(\omega X)$ and $g_i(X) = \omega^i X$ where $\omega$ is a primitive element in $\F_q$. (I.e., $\omega^i = 1$ if and only if $i = 0 \pmod{q-1}$.) The exact reason for this choice remains somewhat a mystery though we will comment on this later.

But before explaining the choice above, we mention a new hitch with the scheme above that some readers may have already noticed. The choice of adding more information to the encoding does have a cost: It reduces the rate of the underlying code. For instance if $f(X)$ is a polynomial of degree $k < n$ and we choose to transmit both $f(a)$ and $g(a)$ for every $a \in S$, then the rate of the code will be at most $1/2$. Dealing with this seems challenging initially but we end up resolving it by a conceptual solution rather than an algebraic one. We won't describe the exact resolution yet, but alert the reader that this is the reason for considering both the ``overlapping'' and the ``non-overlapping'' versions of the codes below.

\paragraph{(Overlapping) FRS Encoder:} The codes we will work with are called Folded Reed-Solomon (FRS). These are codes with a parameter $s$ and have $\Sigma = \F_q^s$ as their alphabets. 

\begin{definition}[(Overlapping) Folded Reed Solomon (FRS) Codes:]
    An overlapping FRS code is specified by integers $s$, $k$, field $\F_q$, primitive element $\omega \in \F_q$, and a set $S \subseteq \F_q$ with $|S|=n$. The encoding function $E = E^{\mathrm FRS}_{s,k,\F_q,\omega,S}$ maps $\F_q^k \to \{f: S\to \Sigma\}$ where $\Sigma = \F_q^s$ as follows: Given a message $m = (m_0,\ldots,m_{k-1}) \in \F_q^k$ and $i \in \{0,\ldots,s-1\}$ let $g_i = g_i^{(m)}$ be the polynomial $g_i(X) = \sum_{j=0}^{k-1} m_j (\omega^i X)^j$. Then $E(m)$ is the function $f:S\to \F_q^s$ given by $f(a) = (g_i(a))_{0\leq i < s}$. 

    If the sets $T_\alpha = \{\alpha\cdot \omega^i | 0 \leq i < s\}$ are disjoint for every pair $\alpha \ne \beta \in S$ then we say that the above code is a non-overlapping FRS code (or simply an FRS code). 
\end{definition}

First note that the Folded Reed-Solomon code is simply a Reed-Solomon code where we've ``bundled'' some of the coordinates together: Specifically we can think of $g_0(X)$ as the message polynomial, and the encoding is simply the evaluations of $g_0$ on the set $S$, and then on the set $\omega\cdot S := \{\omega\cdot \alpha | \alpha \in S\}$, and then on $\omega^2\cdot S$, etc.
However we ``bundle'' the output in a special way by collecting the terms $(g_0(a), g_0(\omega a), \ldots,g_0(\omega^{s-1} a))$ together and calling it a single symbol of the alphabet. 

So how does this bundling affect the decodability? The answer is that bundling and the resulting change in the alphabet from $\F_q$ to $\F_q^s$ changes the power of the channel, which is now restricted in the kinds of errors it can make. (Indeed this is the whole difference between channels over say the alphabet $\F_2$ and $\F_{2^t}$ for large $t$, and this is why channels working with large alphabets are list-decodable from fractions of errors tending to $1$). 

Returning to the codes above, if the bundles used in the codes (we refer to the sets $T_\alpha$ as bundles) overlap then this leads to an overlapping FRS, while if they don't we get a non-overlapping FRS. Overlaps clearly seem wasteful --- they involve transmitting the same information about the message multiple times and this can't be a great use of the channel. However, overlapping FRSes are good to study since 
they admit a non-trivially strong decoder. Later we will see how to reduce the decoding problem for a high rate non-overlapping FRS code to decoding a poor rate overlapping FRS codeword. Together the two steps help us achieve capacity, with judicious choice of parameters. 

\paragraph{Decoding Overlapping FRS:}

We quickly recap the decoding problem. We are given access to a function $r:S \to \Sigma$ and we would like to output all polynomials $p \in \F_q[X]$ of degree less than $k$ such that $|\{a \in S | r(a) = (g_i(a))_{0 \leq i < s} \}| \geq t$, where $g_i(X) = p(\omega^i X)$. (Of course the algorithm is also given the specifications of the code, namely $s,k,\F_q,\omega$ and $S$.) 

\begin{tabbing}
\quad\= {\bf Overlapping FRS List-Decoding Algorithm}\\
\quad Step 1 \quad\= Compute a non-zero polynomial $Q(X,Y_1,\ldots,Y_s) = A_0(X) + \sum_{i=1}^s A_i(X) Y_i$ \\
\quad \quad \quad \quad \quad \quad with $\deg(A_0) < \frac{n-k}{s+1} + k$ and 
$\deg(A_i)\leq \frac{n-k}{s+1}$ for every $1\leq i \leq s$ \\
\quad \quad \quad \quad \quad \quad s.t. $Q(a,r(a))=0$ for all $a\in S$\\
\quad Step 2 \quad\= Find all polynomials $p \in \F_q[X]$ such that $\deg(p)<k$ and $Q(X,g_0(X),\ldots,g_{s-1}(X))=0$, \\
\quad \quad \quad \quad \quad \quad where $g_i(X) = p(\omega^i X)$ for every $1\leq i \leq s$. \\
\quad \quad \quad \quad \quad \quad and output a list of all such polynomials.
\end{tabbing}

The following lemma states the efficacy of the above algorithm.

\begin{lemma}\label{lem:overlap-FRS-decoder}
    The {\bf Overlapping FRS List-Decoding Algorithm} can be implemented to run in time $O(n^3\cdot q^s)$. It outputs a list including all polynomials $p$ of degree less than $k$ such that $|\{a \in S | r(a) = (g_i(a))_{0 \leq i < s} \}| \geq t:= \frac{n-k}{s+1} + k$, where $g_i(X) = p(\omega^i X)$. The list has size at most $O(q^{s})$.
\end{lemma}

Before giving a sketch of the proof of this lemma, we first make some comments. One surprising aspect of the lemma is that it proves a purely information-theoretic fact --- that the output list size is polynomially bounded for when the number of agreements is at least  $t= \frac{n-k}{s+1} + k$ --- a fact that we do not know how to prove without analyzing an efficient algorithm to find the list! (In the case of Reed-Solomon codes, there was an analogous effect --- the efficient list-decoding algorithm proved an upper bound on the list-size. The difference there was the bounds were obtainable by simpler proofs --- and this is not the case here, as far as we know.)

The bound on the number of agreements however is not as good as we may like. Note that the rate of the code is $k/(ns)$ and a ``capacity achieving'' code should have been able to decode from $k/s + o(n)$ agreements, and we are not getting to that regime here. This is an inherent weakness of working with overlapping FRS codes (and the lemma applies to those). We will soon see how to adapt this algorithm for non-overlapping FRS codes via a rate improving reduction, and that will take us to capacity.

\begin{proof}[Proof Sketch:]
The analysis of Step 1 is straightforward given the analysis of similar steps in the RS list-decoding algorithms. Specifically we note that solving for $Q$ requires finding a non-zero solution to a homogeneous linear system, which is guaranteed to have a solution since the number of variables (i.e., the total number of coefficients of $A_0,\ldots,A_s$) is more than $n$, the number of constraints. Thus we turn to Step 2.

Here we first note that if a list satisfying the conditions of Step 2 can be found (efficiently) then it must include every solution $p$ of degree less than $k$ such that $|\{a \in S | r(a) = (g_i(a))_{0 \leq i < s} \}| \geq t:= \frac{n-k}{s+1} + k$, where $g_i(X) = p(\omega^i X)$. 
To see this note that $R(X) := Q(X,g_0(X),\ldots,g_{s-1}(X))$ is a polynomial of degree less than $t$, while every $a$ such that $r(a) = (g_i(a))_{0 \leq i < s}$ is a zero of $R$. Thus $R$ has more zeroes than its degree and must be identically zero. We conclude that $p$ must be included in the output list.

Finally we turn to the algorithmic complexity of Step 2. To this end we note that the simple structure of $Q$ makes the search for $p$ a linear system --- the coefficients of $p$ are the $k$ unknowns and the linear constraints can be derived by computing the expressing the coefficient of $X^i$ in $Q(X,g_0(X),\ldots,g_{s-1}(X))$ as a linear function in the coefficients of $p$ and requiring that this coefficient be $0$. This system is no longer homogeneous, but nevertheless it can be solved in time $O(n^3 q^d)$, where $d$ is the dimension of the space of solutions. 

So to conclude the proof sketch, we only need to upper bound the dimension of the space of solutions to the linear system given by $R(X):= Q(X,g_0(X),\ldots,g_{s-1}(X)) = 0$ (under the restriction that $g_i(X) = p(\omega^i X)$). This step turns out to be a bit tedious and so we won't go into the gory details. But roughly note that this is a ``triangular'' linear system in that we can solve for $p_i(X) := p(X) \pmod X^{i+1}$ by requiring  $R(X) = 0 \pmod X^{i+1}$ for $i=0$ to $k-1$ iteratively. Each $p_{i}$ extends some solution for $p_{i-1}$ by determining the coefficient, say $c_{i}$, of $X^{i}$ in $p_{i}(X)$. The key here is that the condition $R(X) =0 \pmod{X^{i}}$ gives one linear constraint on $c_{i}$ of the form $T(\omega^i) c_{i} = \beta_{i}$ where $\beta_{i+1}$ is determined by the coefficients of $p_{i-1}(X)$ and $Q$, and $T$ is a polynomial of degree at most $s$, determined only by the coefficients of $Q$. When $T(\omega^i)$ is non-zero $p_i$ is uniquely determined from $p_{i-1}$, and when $T(\omega^i)=0$, $c_i$ is not linearly dependent on $p_{i-1}$ and adds to the dimension of the solution space. Since $T$ is a non-zero polynomial of degree at most $s$, we get that $d$ is upper bounded by $s$, thus leading to the claimed bound on the list-size and running time. 
\end{proof}

\paragraph{Rate-improving reduction from non-overlapping FRS to overlapping FRS:}

Our final step is a reduction between the non-overlapping and overlapping cases of the FRS that will improve the rate of the code. More elaborately, given an input to the non-overlapping FRS list-decoding problem for some code of rate $R$ and folding parameter $s$, we will transform it into an input to an overlapping FRS list-decoding problem of much lower rate $R'$ with folding parameter $s'$ while approximately preserving the fraction of errors! Since lower rate codes can typically be decoded from higher fraction of errors, this is a step in a good direction and as we will see, this is exactly what we need to get capacity achieving codes with efficient encoders and decoders.

The idea for the transformation is extremely simple: Given a bundle of size $s$ (of evaluations of some polynomial), we will simply split the bundle into $s-s'+1$ overlapping bundles of size $s'$. Specifically consider the following map, which we call the dilution map, $\phi= \phi^{(s,s')}:\F_q^s\to (\F_q^{s'})^{s-s'+1}$ given by $\phi(b_1,\ldots,b_s) = ((b_i,\ldots,b_{i+s'-1})_{i = 1}^{s-s'+1}$. Now let $\phi_n:(\F_q^s)^n \to (\F_q^{s'})^{n(s-s'+1)}$ be the map that applies $\phi$ to every coordinate of its argument and concatenates the output vectors. Finally let $\phi_n(A) = \{\phi_n(a) | a \in A\}$ for $A \subseteq (\F_q^s)^n$. Then we claim that $\phi_n$ maps a non-overlapping FRS code to an overlapping FRS code with some rate loss. Furthermore we claim that the map preserves relative distances.

\begin{proposition}[Properties of the dilution map]\label{prop:dilute}~
\begin{enumerate}
    \item For every $s'\leq s$, $k$, $\F_q$,$\omega$ and $S \subseteq \F_q$ with $|S|=n$ such that the FRS code $C \subseteq (\F_q^s)^n$ with parameters $s,k,\F_q,\omega,S$ is a non-overlapping FRS code, we have that $C' = \phi_n(C)$ is an overlapping FRS code with parameters $s',k,\F_q,\omega,S'$ where $S' = \{\omega^j \alpha | \alpha \in S, 0 \leq j \leq s-s'\}$. (Note that $n':= |S'| = n(s-s'+1)$.) In particular, if the rate of $C$ is $R$ then the rate of $C'$ is $R' = \frac1{s'}\frac{s}{s-s'+1}R$. 
    \item For every pair of words $x,y \in \Sigma^n$ where $\Sigma = \F_q^s$ we have $\delta(\phi_n(x),\phi_n(y)) \leq \delta(x,y)$. 
\end{enumerate}
\end{proposition}

The proposition follows immediately from the definitions (of FRS codes and the dilution map) and we omit it here. To see how the proposition can be useful, we note that in order to decode a received word $r$ to a nearby codeword in $C$, it suffices to decode $\phi_n(r)$ to a nearby codeword in $C'$. Furthermore by Part (2) the fraction of errors is also preserved, i.e., to list-decode from $p$-fraction errors in $C$ it suffices to list-decode from $p$-fraction of errors in $C'$. Finally to see why decoding $C'$ may be easier, consider the setting where $s'$ is a large constant (independent of $n$) and $s$ is a much larger constant. For such a choice $s/(s-s'+1) \to 1$ and roughly we get that $R' \approx R/(s')$. So $C'$ has much smaller rate than $C$ and hopefully codes of smaller rate can be decoded from larger fractions of error. Indeed this turns out to be sufficient as formalized below.

\begin{theorem}\label{thm:frs}
    For every $\epsilon > 0$ and $R \in [0,1]$ there exists an $s$ and infinitely many $\F_q$ with primitive element $\omega\in\F_q$, $k$, and $S\subseteq \F_q$ with $|S|=n = \Omega_s(q)$ such that the FRS code $C$ with parameters $(s,q,\F_q,\omega,S)$ has rate $R$ and is decodable from $1-R-\epsilon$ fraction errors in polynomial time.
\end{theorem}

\begin{proof}
    Given $\epsilon$ we pick $s' = \lceil 1/\epsilon \rceil$ and $s = (s'+1)(s'-1)$. Given $n$, we pick $q$ to be a prime power satisfying $q-1 \geq sn$. Let $\omega$ be a primitive element in $\F_q$ and let $S = \{\omega^0,\omega^s,\omega^{2s},\ldots,\omega^{(n-1)s}\}$. Finally let $k = \lceil Rsn\rceil$. Let $C$ be the (non-overlapping) FRS code with parameters $(s,k,\F_q,\omega,S)$ and let $C' = \phi_n(C)$. Clearly $C$ has rate at least $R$. We claim that $C$ is list-decodable from $1 - R - \epsilon$ fraction errors in polynomial time in $n$, assuming field arithmetic over $\F_q$ is unit cost. 
    
    To prove the claim we need to list-decode $C$ from $1-R-\epsilon$ fraction of errors. By \Cref{prop:dilute} it suffices to be able to list-decode $C'$ from $1-R-\epsilon$ fraction of errors. We claim that  the {\bf Overlapping FRS List-Decoding Algorithm} achieves this and that this is already implied by \Cref{lem:overlap-FRS-decoder} for the given setting of parameters.
    
    Since $R = k/(sn)$ and $n' = n(s-s'+1)$ we need to list-decode $C'$ from $(1-R-\epsilon)n'$ errors, or $(R+\epsilon)n' = (k/(sn)+\epsilon)n(s-s'+1) = k(s-s'+1)/s + \epsilon n(s-s'+1)$ agreements in $C'$. \Cref{lem:overlap-FRS-decoder} asserts that $C'$ can be decoded from $k + \frac{n'-k}{s'+1}$ agreements. Thus for this algorithm to work for us, it suffices to verify that $k(s-s'+1)/s + \epsilon n(s-s'+1) \geq k + \frac{n'-k}{s'+1}$. To this end we note that we have $k/(s'+1) \geq k(s'-1)/s$ (since $s  \geq (s'+1)(s'-1)$) and $\epsilon n(s-s'+1) \geq n'/(s'+1) = n(s-s'+1)/(s'+1)$ (since $s' \geq 1/\epsilon$.) This yields the desired inequality and concludes the proof.    
\end{proof}

\section{Subsequent Developments and Open Questions}

The results described in the previous section and captured by \Cref{thm:frs} are really nothing short of magical. They say that if you have $k$ pieces of information, they can be encoded so that recovering a ``negligible" amount (formally $o(n)$) of extra pieces correctly in addition to the necessary bare minimum ($k$) suffices to narrow down the original message to a small list, even with arbitrary injection of errors to the remaining $n-k-o(n)$ transmissions (and also allowing the adversary to choose which subset of transmissions to corrupt). 
However once this incredible milestone is achieved, one can ask for even more improvements. We consider some of the natural follow-up questions and describe some of the work in the past two decades that addresses some of the questions (and leave an important one open).

\paragraph{Alphabet size} Reed-Solomon codes were already weak in one sense --- namely that the alphabet size they work with grows at least linearly with the length of the code (i.e., $q \geq n$). Folded Reed-Solomon codes are even weaker in that the alphabet size grows polynomially with the length of the code and in fact, if one wishes to correct from $k + \epsilon n$ agreement, the alphabet size grows as $q \geq n^{1/\epsilon^2}$. A reader may even wonder if these codes are abusing the alphabet size to get to their goal.

It turns out that coding theory already provides good answers to alleviate this concern. Already in 1966 Forney~\cite{forney-thesis} introduced an operation called ``concatenation of codes'' that allowed the use of ``good code'' (of positive rate and distance) over large alphabets to get good codes over small alphabets while preserving polynomial time encoding and decoding. However concatenation in its basic forms suffers a loss in rate and as a result it does not preserve achievement of capacity. Work in the interim, notably the works of Alon, Edmonds and Luby~\cite{AlonEL} and Guruswami and Indyk~\cite{GuruswamiI01} have introduced more sophisticated, graph-theoretic, concatenation methods. Using these methods, the original work of Guruwami and Rudra already prove the following ``constant-alphabet'' version of their capacity achieving results:

\begin{theorem}[\cite{GuRu}]\label{thm:const-alph-cap}
    For every $\epsilon > 0$ and $R \in [0,1]$ there exists an $s$ and alphabet $\Sigma$ (with $|\Sigma| = q = 2^{O(\epsilon^{-2})}$) and polynomial $L$ such that for infinitely many $n$ there exists a code over alphabet $\Sigma$ of length $n$, rate $R$ that is $(1-R-\epsilon,L(n))$-efficiently list decodable.
\end{theorem}

We remark that is also possible to use algebraic geometry to get some capacity achieving codes over constant sized alphabets as shown by Guruswami and Xing~\cite{GuruX22}.

\paragraph{Running time and List-size bound} Next we turn to the running time of the decoding algorithm. As presented (and once all parameters are set), the running time needed to get $\epsilon$-close to capacity (i.e., to decode from $1-R-\epsilon$ fraction errors) is $n^{O(1/\epsilon)}$. It would be desirable to improve this run time to something more like $c_\epsilon n^{c_0}$ where $c_0$ is independent of $\epsilon$ --- i.e., to remove the dependence on $\epsilon$ from the exponent of the running time to just the leading constant. But a major bottleneck towards achieving this result was that the known upper bounds on the list size when decoding from $1-R-\epsilon$ fraction errors have $\epsilon$ in the exponent. Understanding this list size better led to a series of very illuminating results, starting with the work of Kopparty, Ron-Zewi, Saraf and Wootters~\cite{KoppartyRSW}, continuing with Tamo~\cite{Tamo24:FRS} and culminating in the works of Srivastava~\cite{Srivastava} and Chen and Zhang~\cite{ChenZ} that show that Folded Reed-Solomon codes can be designed so as to be $(1-R-\epsilon,1/\epsilon)$-list-decodable. 

These advances return the focus to the running time of the efficient algorithms. While an improvement to the list-size does not necessarily imply an improvement to the running time of the decoding algorithm, the results of \cite{KoppartyRSW,Tamo24:FRS} ended up yielding improvements to the running time also, in particular giving polynomial time algorithms whose running time has a fixed exponent, independent of $\epsilon$, to $n$. A more recent result, by Goyal, Harsha, Kumar, Shankar~\cite{GoyalHKS} ended up giving an essentially optimal result for decoding from $(1-R-\epsilon)$ fraction errors in time $O(n \cdot (\log n)^{c/\epsilon^c})$ for some absolute constant $c$. (Up to the logarithmic factors in $n$ these results are obviously optimal. Some logarithmic factor loss seems inevitable given that we don't even know how to evaluate polynomials of degree $n$ at $n$ places without some logarithmic factor loss in the run time. One can imagine improvements that do not involve $1/\epsilon$ figuring in the exponent of the logarithm though.)


\paragraph{Other capacity achieving codes.} Since the original work~\cite{GuRu} showing that Folded-Reed Solomon codes are capacity-achieving, a few other codes have been discovered with this feature --- all of them being algebraic, or using one of the capacity achieving codes as an ingredient. Perhaps the simplest of these to describe are the derivative codes (also sometimes known as ``univariate multiplicity codes'') --- here the message is again a univariate polynomial and its encoding involves ``bundles" containing the evaluation of the message polynomial and several  ($s-1$) derivatives at various points in the field. Guruswami and Wang~\cite{GuruswamiW13} and Kopparty~\cite{Kopparty:mult} showed that these codes also achieve capacity. Other variants include shifted FRS codes already shown to achieve capacity in \cite{GuRu}, where the bundles contain evaluations of a polynomial over an arithmetic progression (as opposed to a geometric progression); and affine FRS codes shown to achieve capacity in a more recent work of Bhandari, Harsha, Kumar and Sudan~\cite{BhandariHKS}, which are codes that allow combinations of arithmetic and geometric series within a bundle.
One interesting aspect of the work~\cite{BhandariHKS} is that it captures the ``bundling'' operation in an algebraically natural way: They propose thinking of the encoding function as being given by a collection of polynomials $E_1(X),\ldots,E_n(X)$ and the $i$th coordinate of the encoding of a polynomial $P(X)$ is $P(X) \mod E_i(X)$. This generalizes all the codes listed above, and allows properties of the ideals $(E_1(X)), \ldots, (E_n(X))$ to work their way into the analysis of the list-decodability of the underlying codes.(Some readers may note some similarity with the use of the Chinese Remainder theorem here --- indeed this theorem and the abstract theory behind it do form the basis of much of this generalization.) However even the framework in \cite{BhandariHKS} turns out to be not the most general. In \cite{BermanST}, Berman, Shany and Tamo introduce a new family of codes based on bivariate polynomials as the message space that end up being capacity achievable with algorithms and analysis along the same lines as FRS codes, but do not fall in the framework of \cite{BhandariHKS}. This in turn has induced a further generalization of \cite{BhandariHKS} to ``Bivariate Linear Operator'' codes in the work of Putterman and Zaripov~\cite{PuttermanZ} that ends up being the most general known framework capturing all known basic capacity achieving codes. (We stress these do not capture the algebraic-geometry based codes yet.)

While algebraic codes do form the essence of all known capacity-achieving codes, we do have an interesting collection of codes that do not fall entirely in this framework of ``ideal-theoretic'' codes. Instead they use any capacity achieving code as an ingredient, but then enhance it with non-algebraic tools to get new features that the underlying algebraic code does not possess. Indeed the ``constant-alphabet'' capacity achieving codes of \Cref{thm:const-alph-cap} are already a good example of such an enhanced code. An elegant example of such an enhancement is in the work of Hemenway, Ron-Zewi and Wootters~\cite{HemenwayRW} who show that tensor-products of codes retain the capacity achieving feature. (The tensor product operation, due to Elias~\cite{Elias-tensor},  takes two codes $C_1 \subseteq \F_q^{n_1}$ and $C_2 \subseteq \F_q^{n_2}$ and produces a new code $C_1 \otimes C_2 \subseteq \F_q^{n_1 \times n_2}$ whose codewords can be viewed as $n_1 \times n_2$ matrices and consists of all matrices whose columns are codewords of $C_1$ and rows are codewords of $C_2$.) This general property, applied to appropriate codes leads to nice properties like fast running time of decoding algorithm and ``locality'' (a notion that we will not elaborate on here).

\paragraph{Open Question: Capacity achieving binary codes}

While FRS codes and variants do manage to achieve capacity over sufficiently large alphabets, when it comes to very small alphabets they do not do so, even with the best known concatenation techniques. The setting of the binary alphabet ($q=2$) in particular highlights the gap. Here it is well-known that if the goal is to correct from $p$-fraction errors for with polynomial sized lists, then the best known codes have rate approaching $1 - h(p)$ for $p \leq 1/2$ where $h(p) = - p\log_2 p - (1-p) \log_2 (1-p)$ is the binary entropy function. Thus the capacity of list-decoding over the binary alphabets is $1 - h(p)$. However no explicit codes approaching capacity are known and we don't know if there exist such codes with polynomial time encoding and decoding algorithms. 

While it would be desirable to get such codes for all choices of $p$ even the setting of $p \to 1/2$ illustrates the gap and does so qualitatively. If $p = \frac12 - \gamma$ and $\gamma \to 0$, then the capacity $1 - h(p)$ takes the form $\widetilde{\Theta}(\gamma^{2})$ (where the $\widetilde{\Theta}$ notation hides factors that grow logarithmically in its argument). The best known codes, also from \cite{GuRu}, achieve a rate of $\Omega(\gamma^4)$. Despite a lot of effort aimed at improving this rate over the past two decades the exponent has held steady. We hope this setting remains a focus of future research. 

\paragraph{An Addendum: New capacity-achieving codes}
During the editing stages of this article the author was informed about an exciting sequence of recent results that builds capacity-achieving codes that can be list-decoded in nearly linear time! These codes are constructed by completely different methods than those described in this paper. We refer to the reader to the papers~\cite{Tuls1,Tuls2} for further details. 

\section*{Acknowledgments} 

I would like to thank Kumar Murty for encouraging me to write this article. The article benefits from discussions over the years with my colleagues and collaborators and I'd like to thank
Siddharth Bhandari,
Oded Goldreich,
Venkatesan Guruswami,
Prahladh Harsha,
Swastik Kopparty, 
Mrinal Kumar, 
Atri Rudra,
Ramprasad Saptharishi, 
Shubhangi Saraf, 
Srikanth Srinivasan,
Salil Vadhan, and Avi Wigderson
for improving my understanding of the results covered here. 
Thanks to Prashanth Amireddy, Oded Goldreich, Cassandra Marcussen and Aaron (Louie) Putterman for proofreading this article and catching many errors.
Thanks to an anonymous reviewer for their careful review that caught several errors and point out inconsistencies.



\end{document}